\documentclass[conference]{IEEEtran}
\usepackage{cite}
\usepackage{amsmath,amssymb,amsfonts}
\usepackage[english]{babel}
\usepackage{amsthm}
\usepackage{amssymb}
\usepackage{cleveref}
\theoremstyle{plain}

\newtheorem{proposition}{Proposition}
\usepackage{algorithmic}
\usepackage{graphicx}
\usepackage{textcomp}
\usepackage{xcolor}
\usepackage{comment}
\def\BibTeX{{\rm B\kern-.05em{\sc i\kern-.025em b}\kern-.08em
    T\kern-.1667em\lower.7ex\hbox{E}\kern-.125emX}}
    
\def\metric{AoD}
\newcommand{\jpcol}[1]{{\color{black} #1}}

\begin{document}

\title{Minimizing Age of Detection for a Markov Source over a Lossy Channel
}

\author{
    \IEEEauthorblockN{Shivang Garde\IEEEauthorrefmark{1}, Jaya Prakash Champati\IEEEauthorrefmark{2}, Arpan Chattopadhyay\IEEEauthorrefmark{1}}
    \IEEEauthorblockA{\IEEEauthorrefmark{1} Department of Electrical Engineering, Indian Institute of Technology, Delhi, India}
    \IEEEauthorblockA{\IEEEauthorrefmark{2} IMDEA Networks Institute, Madrid, Spain}
    \IEEEauthorblockA{$\{$ee1180501, arpan.chattopadhyay$\}$@ee.iitd.ac.in, jaya.champati@imdea.org}
}

\maketitle

%\author{Shivang Garde, Jaya Prakash Champati, and Arpan Chattopadhyay\\
%$\{$ee1180501, arpanc $\}$@ee.iitd.ac.in, jaya.champati@imdea.org.
%\thanks{Shivang Garde and Arpan Chattopadhyay are with the Department of Electrical Engineering, Indian Institute of Technology, Delhi, India. Jaya Prakash Champati is with IMDEA Networks Institute, Madird, Spain.  E-mail: jaya.champati@imdea.org }

%\thanks{Arpan Chattopadhyay acknowledges support via the professional development fund and professional development  allowance from IIT Delhi, grant no.GP/2021/ISSC/022 from I-Hub Foundation for Cobotics and grant no.CRG/2022/003707 from Science and Engineering Research Board (SERB),India.}

\begin{abstract}
Monitoring a process/phenomenon of specific interest is prevalent in Cyber-Physical Systems (CPS), remote healthcare, smart buildings, intelligent transport, industry 4.0, etc. A key building block of the monitoring system is a sensor sampling the process and communicating the status updates to a monitor for detecting events of interest. Measuring the freshness of the status updates is essential for the timely detection of events, and it has received significant research interest in recent times. In this paper, we propose a new freshness metric, Age of Detection (\metric), for monitoring the state transitions of a Discrete Time Markov Chain (DTMC) source over a lossy wireless channel. We consider the pull model where the sensor samples DTMC state whenever the monitor requests a status update. 
%The status updates are transmitted over a lossy channel. For this system, we propose a new freshness metric  that measures the staleness of the observed states at the monitor and is tractable without assuming the knowledge of the exact DTMC states in between sampling instants. 
We formulate a Constrained Markov Decision Problem (CMDP) for optimizing the \metric{} subject to a constraint on the average sampling frequency and solve it using the Lagrangian MDP formulation and Relative Value Iteration (RVI) algorithm. Our numerical results show interesting trade-offs between  \metric, sampling frequency, and transmission success probability. Further, the AoD minimizing policy provides a lower estimation error than the Age of Information (AoI) minimizing policy, thus demonstrating the utility of AoD for monitoring DTMC sources.
\end{abstract}

%\begin{IEEEkeywords}
%component, formatting, style, styling, insert
%\end{IEEEkeywords}

\section{Introduction}
%The ability to detect the occurrence of an event while remotely monitoring an information source or a process of interest is crucial for a wide range of applications. Numerous applications, ranging from unmanned aircraft and self-driving cars to the Internet of Things, depend heavily on the transmission of time-sensitive data. On the World Wide Web, for example, a web crawling application is responsible for downloading distant web pages to a local database and identifying events that occur when the remote web page is modified. In many of these applications, the freshness of data is critical.

Remote monitoring of an information source or a process of interest for detecting the occurrence of an event, e.g., the process exceeding a threshold, is crucial for a wide range of systems that include unmanned aircraft, self-driving cars, smart buildings, and Industry 4.0. For timely detection of the events, the monitor requires fresh status updates which in turn requires a sensor to sample the process as frequently as possible. However, sampling at such a high frequency could result in high energy consumption and a large bandwidth requirement for transmitting the samples to the monitor. On the other hand, sampling at a lower frequency causes staleness in event detection. Also, in practice, the communication channel is not ideal, i.e., the samples transmitted through the channel may not always be successfully received by the receiver. 
%So, a lower sampling frequency may result in the sensor sending fewer event occurrence instances successfully because of the channel losses. 
In this context, we address the following question: for a discrete-time remote monitoring system with a lossy communication channel, 
\textcolor{black}{at every time slot, given the previously sampled state of the source, whether to sample or not to optimize both frequency of sampling and staleness.} 

We address the above question for sources modeled by finite-state DTMC, which received considerable attention recently \cite{Kam2018,Kam2020,Yutao2021,Maatouk2020,ChampatiTCOM2022}. The events we seek to identify are the new state transitions in the DTMC. The motivation for studying the DTMC source is that it can serve as an abstract model to describe the states of a plant or a process in industries. For example, we may classify the states of a plant into “good” and “bad” and study a 2-state DTMC by deriving the transition probabilities from the history of observations made on the plant. In this paper, we consider that the DTMC states are unknown to the sensor between sampling instants. This is relevant in monitoring applications such as machine fault detection\cite{Wade2015}, where an IoT sensor wakes up periodically and samples the machine state (vibrations, rotation speed, etc.) and does not know the state of the machine in between the sampling instants.
%The trade-off difficulties we explore have not been addressed in the literature to our knowledge.

%In order to quantify the staleness of detecting an event, we first need to choose a suitable metric that possesses two basic properties: 1) it needs to take into account the freshest observed DTMC state, and 2) it needs to be computationally tractable even without the knowledge of the state transitions in between sampling instants. 

Quantifying the staleness of status updates using the Age of Information (AoI) metric was extensively studied for a wide variety of queuing and wireless networking systems in the past decade; see for example~\cite{kaul_2011a,kaul_2012a,Chen2016,Costa_2016,Najm_2016,Najm_2017,Yoshiaki2018,Soysal2019,YatesSurvey2020,SunBook2020,KostaBook2017}. AoI is defined as the time elapsed since the generation of the freshest update available at the destination. 
% Cite all works including AoS, etc. Then quickly summarize works other than AoII and Ao2I. Then explain AoII and papers that works with it. Then explain Ao2I and papers that work with it.
AoI and any cost function of AoI, e.g., see \cite{Tripathi2019}, only capture the timing of the information update and do not account for the semantics or the dynamics of the information source. Thus, there has been significant research interest in studying new freshness metrics that account for variations or states of the source. Some early studies proposed freshness metrics based on mutual information in~\cite{SunBenjamin2018} and conditional entropy in~\cite{Feng2020}. More recently, an entropy-based freshness metric was also studied in \cite{Chen2024}, where the authors consider the problem of selecting a subset of Markov sources for transmitting over unreliable shared channels. The problem is solved using the restless MAB framework. Unlike their metric, we define a new freshness metric, Age of Detection (AoD), as a per-slot cost in an MDP and solve a CMDP to minimize AoD subject to a sampling frequency constraint.
%However, computing the mutual information between the source and the receiver or the conditional entropy is challenging. \cite{Ferguson2023}
%While the authors in~\cite{SunBenjamin2018} proposed a freshness metric using mutual information of the source, the authors in~\cite{Feng2020} proposed a metric using conditional entropy. 

Another semantics-based freshness metric \textit{Age of Incorrect Information} (AoII) proposed in \cite{Maatouk2020} received significant research attention. It is defined as a product of two functions: one is an increasing function of time until the next update is received, and the other is a general penalty function for mismatch in the states at the source and the receiver. Different instantiations of AoII were studied in \cite{Maatouk2020,Kam2018,Kam2020,Yutao2021} for estimating the states of a DTMC source. However, these works consider a parameterized symmetric DTMC for analytical tractability, where the transition probability matrix has the following special structure. For any state, the self-transition probability is $p_t$, and the transition probability to any other state is $p_R$. Consequently, these values have to satisfy $p_R + (N -1)p_t = 1$, where $N$ is the number of states. %\textcolor{red}{Is N the no of states? Why such a specia model?}

In contrast to the above works, in our previous work~\cite{ChampatiTCOM2022}, we studied the $N$-state ergodic DTMC without further assumptions on the structure of the state transition probability matrix and considered that the sensor only knows the state of the DTMC when it samples. We proposed the \textit{age penalty} metric, defined as the time elapsed since the first DTMC state change occurred since the last sampling instant. 
%Age penalty the scenario where  both the properties mentioned above. 
We analyzed the age penalty metric for the scenario where the communication between the sensor and the monitor is either negligible or takes exactly one slot. The age penalty metric was independently studied under Age of Outdated Information (Ao$^2$I) in~\cite{Liu2022}. However, \cite{ChampatiTCOM2022} and \cite{Liu2022} assumed a deterministic communication delay of one slot for any status update transmission. In this paper, we relax this assumption by considering a lossy wireless channel. The random delays induced by the losses cannot be accounted for in the age penalty metric, which motivates the AoD definition. \jpcol{In a slot, if the receiver sends no request for a new update or if there is a transmission failure, \metric{} induces a cost that is a function of the time elapsed since the generation time of the freshest update at the receiver and the time elapsed since the latest request sent.}  A salient feature of \metric{} is that it is tractable, and we show that it reduces to the age penalty metric under the assumption that the communication channel has no losses. 

Our main contributions are summarized below:
\begin{itemize}
    \item We propose \metric{} for quantifying the freshness of the information in monitoring the state transitions of an  $N$-state ergodic DTMC over a lossy wireless channel. 
    %If there are no losses in the channel, then we show that \metric is equivalent to an existing \textit{age penalty} metric proposed in \cite{ChampatiTCOM2022}, where the communication channel between the sensor and the receiver is ignored.
    \item We study the trade-off between \metric{} and sampling frequency by formulating a Constrained Markov Decision Problem (CMDP) in which we minimize \metric{} subject to a constraint on the sampling frequency. 
    %We solve the CMDP using a Lagrangian MDP formulation and use RVI to solve the Lagrangian MDP. 
    \item Our numerical results demonstrate that, for a fixed sampling frequency constraint, \metric{} exponentially decreases with an increase in transmission success probability. We also evaluate the performance of the CMDP solution by comparing it with the \metric{} achieved by the optimal clairvoyant policy, which assumes the complete knowledge of the state transitions in the DTMC and samples the DTMC exactly when a new state transition happens. 
    \item Finally, using simulation, we demonstrate that the estimation error under the \metric{} minimizing policy is lower than that of the AoI minimizing policy. More interestingly, under some settings, the \metric{} minimizing policy achieves lower error at a lower sampling frequency. 
\end{itemize}

\section{Age of Detection: Definition and Analysis}
\begin{figure}
	\centering
	\includegraphics[scale=0.30]
	{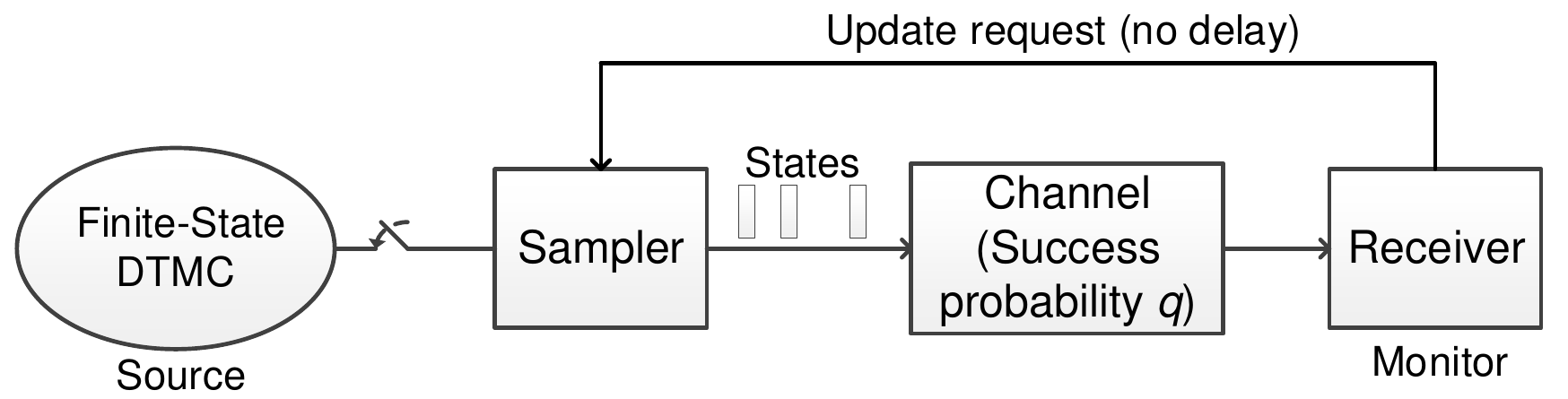}
	\caption{The pull model used for sampling information for a DTMC source's state transitions.}
	\label{fig:model}
\end{figure}
%{Sampling_process.drawio (1).png}
\subsection{System Model}
We consider the system shown in \Cref{fig:model}. The source is a $N$ state ergodic DTMC. The set of states is given by $\mathcal{X} = \{1,2,\ldots, N\}$. Time is slotted with the slot duration equal to the unit of time of the DTMC. The monitor/receiver aims to detect the state transitions of the DTMC. We study the pull model, where at the start of a slot, the receiver has to send a request to a sensor at the source to obtain the state of DTMC in that slot. We assume this request is received instantaneously by the sensor. When the sensor receives the request, it samples the DTMC's state and transmits it over a lossy channel. The success probability of receiving the transmitted sample/state within a time slot is $q$.  If the transmission is unsuccessful, the sensor keeps re-transmitting the state until either a new request is received; in this case, it samples the latest state and transmits it, or the most recent transmission is successful. When the monitor receives a state, it immediately transmits 

In \Cref{fig:System}, we show sample generation and reception time slots. In time slot $t$, $T_{G}(t)$ denotes the generation time of the latest state that is available at the monitor, and $T_R(t)$ denotes its reception time. Note that some generated states are lost in scenarios where a state is transmitted unsuccessfully in the current slot, and the sensor receives a request from the monitor at the start of the next slot. We use $T_g(t)$ to denote the time of the latest request. Further, we define $\tau_1 = T_g(t) - T_G(t)$ and $\tau_2 = t - T_g(t)$. Note that, at time $t$, a freshness metric that measures the staleness of the state at the monitor requires keeping track of the DTMC states change in the time elapsed since $T_G(t)$, given by $\tau_1+\tau_2$. Toward this end, we define AoD as the cost incurred incurred per slot in a Markov Decision Process (MDP). We define the elements of the MDP below.

\begin{figure}[h]
	\centering
	\includegraphics[scale=0.4]{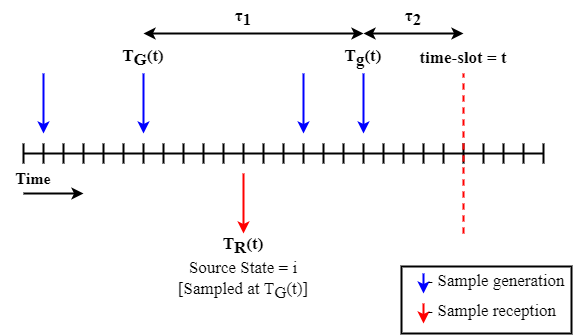}
	\caption{A pictorial representation showing the sample generation and reception instances over a short period of time.}
	\label{fig:System}
\end{figure}

\subsection{Elements of MDP}
\subsubsection{\textbf{State space}}
%We assume that the sensor knows the state of the source only when it samples on the arrival of the update request. Let $S_{t}$ denote the state of the Markov source known to the receiver at time $t$, where $S_{t} \in \mathcal{X} = \{1, 2, . . . ., N\}$. We denote the time of the latest request by $T_{g}(t)$. It can be noted that $T_g(t) > T_R(t) > T_G(t)$ holds true always because $T_g(t)$ represents the time of the latest request and it gets updated and replaced by a new measurement each time we request for the source's state information. 
The state of the MDP, denoted by $S_t$, is realized by the tuple $(\tau_{1}, \tau_{2}, i, j)$, where $i \in \mathcal{X}$ and $j \in \mathcal{X}$ denote the states of the DTMC in time slots $T_{G}(t)$ and $T_g(t)$, respectively. The reason that we keep track of $\tau_1$ and $\tau_2$ separately instead of their sum will be clarified shortly.

\subsubsection{\textbf{Action space}}
The action space $\mathcal{U}=\{0,1\}$, where in time slot $t$, the action $u_t = 1$ denote request sent by the monitor, and $u_t = 0$, otherwise. 
%We assume that the monitor immediately requests an update when a successful reception occurs. This is 

\subsubsection{\textbf{Transition Probabilities}}
Let the state of the MDP at time $t$ be $S_{t} = (\tau_{1}, \tau_{2}, i, j)$, then the next state will depend on whether the monitor requests for an update and successful/unsuccessful transmission in the current slot. For the time slots in which the monitor does not send a request, i.e., $u_t=0$ and $T_g(t)$ does not change, the transition probabilities are given by
%If we do not request and the ongoing transmission fails then, the next state is the same as the previous one. Therefore, $P(S_{t+1} = (\tau_{1}, \tau_{2}+1, i, j)~|S_{t} = (\tau_{1}, \tau_{2}, i, j), u_{t} = 0) = (1-q)$ and if it succeeds then $S_{t+1} = (\tau_{2}, 1, j, j')$, where $j'$ is the state observed at $t$, with probability $p_{jj'}^{(\tau_{2})}q$. If we choose to sample and transmission is unsuccessful then, $S_{t+1} =  (\tau_{1}+\tau_{2}, 1, i, j')$ with probability $p_{jj'}^{(\tau_{2})}(1-q)$ and if successful then $S_{t+1} = (1, 1, j', j')$ with probability $p_{jj'}^{(\tau_{2})}q$.
\begin{align}
   &P(S_{t+1} = (\tau_{1}, \tau_{2}+1, i, j)~|S_{t},u_{t} = 0) = 1-q \label{eq1:transProb}\\
    &P(S_{t+1} = (0, \tau_2+1, j, j)~|S_{t},u_{t} = 0) = q. \label{eq2:transProb}
\end{align}
The transition probability in \eqref{eq1:transProb} corresponds to an unsuccessful transmission with probability $1-q$. Recall that $\tau_1 = T_g(t) - T_G(t)$ and $\tau_2 = t - T_g(t)$. Therefore, in \eqref{eq1:transProb}, the only change in the state possible is $\tau_2$ is increased. Similarly, in \eqref{eq2:transProb}, when the transmission is successful with probability $q$, then the state $j$ at time $T_g(t)$ is received by the monitor, $T_G(t)$ is set to $T_g(t)$ resulting in $\tau_1 = 0$ $\tau_2$, and $\tau_2$ is increased by $1$.

When $u_t = 1$, we have $T_g(t) = t$. Let $j'$ denote the state sampled in $t$, then the transition probabilities are
\begin{align}
    &P(S_{t+1} \!=\! (\tau_{1}+ \tau_{2},1, i, j')~|S_{t},u_{t} \!=\! 1) \!=\! p_{ij'}^{(\tau_{1}+\tau_{2})}(1-q) \label{eq3:transProb}\\
      &P(S_{t+1} = (0, 1, j', j')~|S_{t},u_{t} = 1) = p^{(\tau_1+\tau_2)}_{ij'}q. \label{eq4:transProb}
\end{align}
Note that since the new state $j'$ is observed in time slot $t$, the probability of transition to state $j'$ from state $i$ is given by $p^{(\tau_1+\tau_2)}_{ij'}$. If the transmission in slot $t$ is unsuccessful, then $T_g(t) - T_G(t) = \tau_1+\tau_2$ as given in \eqref{eq3:transProb}, otherwise $T_g(t) - T_G(t) = 0$ as given in \eqref{eq4:transProb}.
%At time $t$, the value of $\tau_2$ is known from the state. So, we can express the probability in terms of that. Also, the value of $T_g$ is known and is fixed at time t, because we know when we last requested for update.
%\textcolor{red}{Use caligraphic font only for sets, and normal mode for variables.}

\subsubsection{\textbf{\metric}}
We define \metric, denoted by $c(S_{t}, u_{t})$, as a per-slot cost which is given below.
%\begin{comment}
\begin{multline*}
c(S_{t}, u_{t}) = u_{t}(1-q)\sum_{j=1, j\neq i}^{N}p_{ij}^{(\tau_{1})}(1-p_{jj}^{\tau_{2}-1}) \\ + (1-u_{t})\sum_{j=1, j\neq i}^{N}p_{ij}^{(\tau_{1})}(1-p_{jj}^{\tau_{2}-1}).   
\end{multline*}
%\end{comment}
\begin{comment}
\begin{align*}
c(S_{t}, u_{t}) = u_{t}(1-q)\sum_{j=1, j\neq i}^{N}p_{ij}^{(\tau_{1})}(1-p_{jj}^{\tau_{2}-1})  + (1-u_{t})\sum_{j=1, j\neq i}^{N}p_{ij}^{(\tau_{1})}(1-p_{jj}^{\tau_{2}-1}).   
\end{align*}
\end{comment}
\metric{} defined above can be simplified as
\begin{equation}\label{eqn: cost function}
c(S_{t}, u_{t}) = (1-qu_{t})\left(1-p_{ii}^{(\tau_{1})}-\sum_{j=1, j\neq i}^{N}p_{ij}^{(\tau_{1})}p_{jj}^{\tau_{2}-1} \right)    
\end{equation} 
\textcolor{black}{The intuition behind defining \metric{} is to capture the error in the estimation of the state along with maintaining the freshness of information. The term $p_{ij}^{(\tau_{1})}(1-p_{jj}^{\tau_{2}-1})$ captures the probability that the source attained state $j$ at latest sampling time $T_g(t)$ but left the state $j$ before current time $t$. In such a case, we must penalize the system for wrongly assuming the state as $j$. We choose the penalty-per-slot to be the probability of error in state estimation. Also, the only case where the system should not be penalized is when $q=1$ and $u=1$, for which AoD is zero (cf. \eqref{eqn: cost function}). This is consistent with our understanding that when the decision is to sample, and the transmission is successful, then the estimation error is minimal, and the information is as fresh as it can be.}

To further motivate the definition of AoD in \eqref{eqn: cost function}, we claim that AoD extends the metric age penalty \cite{ChampatiTCOM2022} (or Ao$^2$I \cite{Liu2022}) for wireless channels with losses. If the wireless channel has no loss, then the average AoD is equal to the average age penalty. We state this in the following proposition.
\begin{proposition}
If the wireless channel has no loss, i.e., $q=1$, then the average AoD equals the average age penalty.
\end{proposition}
\begin{proof}
Given $q=1$, we have,
\begin{equation*}
c(S_{t}, u_{t}) = (1-u_{t})\sum_{j=1, j\neq i}^{N}p_{ij}^{(\tau_{1})}(1-p_{jj}^{\tau_{2}-1}).  
\end{equation*}
Further, if $q=1$, we have $T_g(t) = T_G(t)$, which implies $\tau_1 = 0$, and there is exactly a unit delay in transmission. So, if we look at the period between any two receptions, there will be no new sampling instances; hence, $u_{t} = 0$ for all time slots between two receptions. 
%Also, the state at last reception $i$ is the only known information about the source, i.e., there is no instance in between where the state can be known, hence, no meaning in taking the expected value of cost over any state in between. 
Also, any state $j$ sampled is always received, implying that $i = j$, and we do not have the expectation taken with respect to any other state in the cost definition of AoD.
Therefore, for  $q=1$, \metric{} for each slot in between two receptions is given by
\begin{equation*}
c(S_{t}, u_{t}=0) = 1-p_{jj}^{\tau_{2}-1}.   
\end{equation*}
Let $\tau$ denote the number of time slots between two receptions.
Summing the above equation over values of $\tau_2$ from $1$ to $\tau$, because after each reception $\tau_2$ starts from $1$, we get, 
\begin{align*}
\sum_{\tau_2=1}^{\tau}(1-p_{jj}^{\tau_2-1}) &= \tau - \sum_{\tau_2=1}^{\tau}p_{jj}^{\tau_2-1}\\
&=\tau - \left(\frac{1-p_{jj}^{\tau}}{1-p_{jj}}\right).   
\end{align*}
The expression above, which is the sum of the \metric{} between two receptions, is exactly the age penalty metric (cf. equation (4) from \cite{ChampatiTCOM2022}). Given this result, we use the Renewal-Reward Theorem to show that the average AoD equals the average age penalty.
\end{proof}

We aim to minimize the time-averaged expected \metric{} under the sampling frequency constraint. To this end, given the elements of MDP, we formulate the following problem:
\begin{equation}\label{eqn: optimization problem}
\begin{aligned}
\min_{\pi \in \Pi} \quad & \lim_{T \to \infty} \mathbb{E}^{\pi}\left[\frac{1}{T}\cdot \sum_{t=0}^{T-1} c_{t}(S_{t}, u_{t})  \right]\\
\textrm{subject to} \quad & \lim_{T \to \infty} \mathbb{E}^{\pi}\left[\frac{1}{T}\cdot \sum_{t=0}^{T-1} u_{t}  \right] \leq \nu \\
\end{aligned}
\end{equation}
This is a \textit{Constrained Markov Decision Problem (CMDP)}. 

\subsection{Lagrangian MDP}
We use the Lagrangian relaxation approach to convert the CMDP to an unconstrained MDP. 
%We formulate the Lagrangian MDP with parameter $\lambda \geq 0$ and compute the optimal deterministic policy for the Lagrangian MDP. We recover the optimal policy for \eqref{eqn: optimization problem} by randomizing between two such optimal deterministic policies corresponding to two Lagrangian MDPs obtained at two chosen $\lamda$ values.
At time slot $t$, if the state of the MDP is $S_{t}$ and the action chosen is $u_{t}$, then the relaxed per-stage cost is defined as
\[
c_{t}^{\lambda}(S_{t}, u_{t}) = c_{t} + \lambda u_{t}, 
\]
where $\lambda \geq 0$ is a Lagrange multiplier. 
Then, the Lagrangian MDP for the constrained optimization problem will be, for $\lambda \geq 0$, 
\begin{equation}\label{eqn: lagrangian opt problem}
\begin{aligned}
\min_{\pi \in \Pi} \quad & \lim_{T \to \infty} \mathbb{E}^{\pi}\left[\frac{1}{T}\cdot \sum_{t=0}^{T-1} c_{t}^{\lambda}(S_{t}, u_{t})  \right].\\
\end{aligned}
\end{equation}
% We leave the $\lambda \nu$ term because it is a constant and doesn't affect the solution of the Lagrangian MDP. Let $\pi_{\lambda}^{*}$ be the optimal policy then if
% \[
% g(\lambda) = \lim_{T \to \infty} \mathbb{E}^{\pi_{\lambda}^{*}}\left[\frac{1}{T}\cdot \sum_{t=0}^{T-1} c_{t}^{\lambda}(S_{t}, u_{t})  \right] - \lambda \nu
% \]
Let $\pi_{\lambda}^{*}$ denote the optimal Lagrangian MDP policy for a given $\lambda$. Then, from the Lagrangian duality result the minimum average \metric{} in the CMDP is given by
\begin{equation}\label{eqn: lambda}
\max_{\lambda \geq 0} \left\{\lim_{T \to \infty} \mathbb{E}^{\pi_{\lambda}^{*}}\left[\frac{1}{T}\cdot \sum_{t=0}^{T-1} c_{t}^{\lambda}(S_{t}, u_{t})  \right] - \lambda \nu \right\}.
\end{equation}
%\textcolor{red}{Function R not defined}
\textcolor{black}{It is well-known that if there exists a $\lambda^* >0$,  and a policy $\pi_{\lambda^*}^*$ which is optimal for the relaxed MDP under $\lambda^*$ and satisfies the constraint in the CMDP with equality, then $\pi_{\lambda^*}^*$ is optimal for the CMDP as well. It is also well-known that the mean number of requests under $\pi_{\lambda}^*$ decreases in $\lambda$. Hence, if the constraint can not be met with equality under any deterministic $\pi_{\lambda}^*$ for any $\lambda>0$, then we need to randomize between $\pi_{\lambda^*+}^*$ and $\pi_{\lambda^*-}^*$ for a suitable $\lambda^*$ to meet the constraint in the CMDP \eqref{eqn: optimization problem}; see \cite{4048912,BEUTLER1985236}. Using these results, we solve the problem numerically and obtain the optimal policies and the average \metric{}.}

\section{Numerical Results}

% We use the \textit{Relative Value Iteration (RVI) Algorithm} to obtain the optimal policy and average age penalty. 
% Let $V(S)$ denote the value function for state $S = (\tau_{1}. \tau_{2}, i, j)$. Using Bellman equation, we have,
% \begin{equation*}
%     V(S) =\min_{u \in \{0, 1\}}[c^{\lambda}(S, u) + \mathbb{E}_{\mathcal{S'}}(V(S))]
% \end{equation*}
% Let $J(S) = V(S)-V(S^{(0)})$, where $S^{(0)}$ is the initial state, which we take as $(1,1,0,0)$, we have the update equation for the $k$th iteration as,
% \begin{multline*}
%      J_{k}(S) =\min \{c_{k}^{\lambda}(S, 1) + q\sum_{j'}p_{jj'}^{(\tau_{2})}J_{k-1}(1,1,j',j') \\ +(1-q)\sum_{j'}p_{jj'}^{(\tau_{2})}J_{k-1}(\tau_1+\tau_2,1,i,j'), \\
%      c_{k}^{\lambda}(S, 0) + q\sum_{j'}p_{jj'}^{(\tau_{2})}J_{k-1}(\tau_2,1,j,j') \\ +(1-q)J_{k-1}(\tau_1, \tau_2 + 1, i, j) \} - g_k  ~~\forall ~S \neq S^{(0)}
% \end{multline*}
%   where, 
% \begin{equation*}
%     g_k = \min_{u \in \{0, 1\}}[c^{\lambda}(S^{(0)}, u) + \mathbb{E}_{S'^{(0)}}(V(S'^{(0)}))]
% \end{equation*}
% After the RVI converges, the value of $g$ is the optimal average value of the $J$ function. We, then, use this value to find the optimal value of the Lagrange multiplier and use that to obtain the randomized mixture of two pure policies as our optimal policy. 

\subsubsection{Variation with transition probability} \textcolor{black}{For simulations, we consider a two-state Markov source with state transition probabilities $p_{01}$ and $p_{10}$ between states $0$ and $1$. We choose the default values for transmission success probability $q = 0.8$ and the constraint on sampling frequency $\nu = 0.1$. We limit the values of $\tau_1$ and $\tau_2$ to be in the range $1$ to $20$, i.e., when these variables cross the threshold, we assume a successful reception which resets the value of that variable and obtain the policy for different values of state transition probabilities. This assumption is valid since the probability of a successful transmission, when the system has re-transmitted as many as $20$ times, is close to $1$ for the chosen range of values for $q$. Also, this ensures that the simulations run in a reasonable amount of time. We obtain $\lambda^* = 0.3$ by solving \cref{eqn: lambda} and take two values, one slightly bigger ($\lambda^*+ = 0.3001$), other slightly lower ($\lambda^*- = 0.2999$) than $\lambda^*$ to obtain two pure policies $\pi_{{\lambda}^*-}$ and $\pi_{{\lambda}^*+}$, and compute the randomization factor $= 0.428$, which results in the optimal policy $\pi^*$. We obtain the following results for the optimal $(\tau_1,\tau_2)$, given the states $(i, j) = (0, 0)$, i.e., the latest received state and the latest sampled state is $0$.}
%\textcolor{red}{Fig 3-5: Decision should depend on the entire state vector and not only on tau!!!} \textcolor{blue}{These observations are for when states (i, j) are (0, 0).}

% \begin{itemize}
    % \item {\underline{Case-1}}: $p_{01} = 0.02$, $p_{10} = 0.01$. For these values of probabilities the optimal decision is to sample for states with values of $\tau_1$ and $\tau_2$ shown in blue color in \Cref{fig:p=0.02}.

\begin{figure}
    \centering
    \includegraphics[scale=0.53]{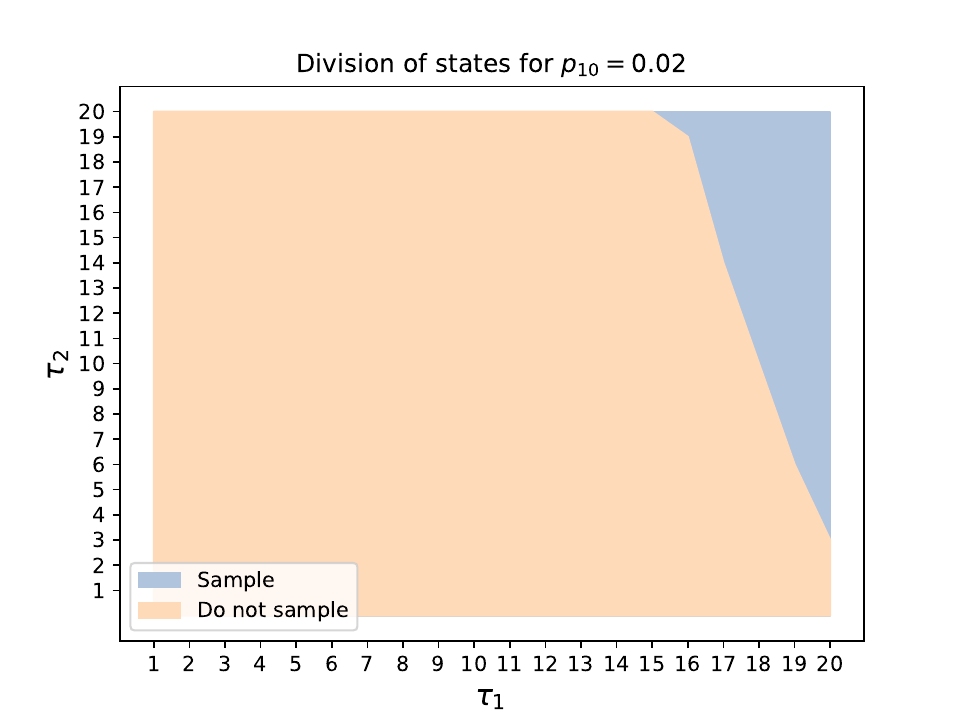}
    \caption{Division of states on the basis of optimal decision for each state for $p_{01}=0.02$, $p_{10}=0.01$, $q = 0.8$ and $\nu=0.1$.}
    \label{fig:p=0.02}
\end{figure}

    % \item {\underline{Case-2}}: $p_{01} = 0.03$, $p_{10} = 0.01$. For these values of probabilities the optimal decision is to sample for states with values of $\tau_1$ and $\tau_2$ shown in blue color in \Cref{fig:p=0.03}.

\begin{figure}
    \centering
    \includegraphics[scale=0.53]{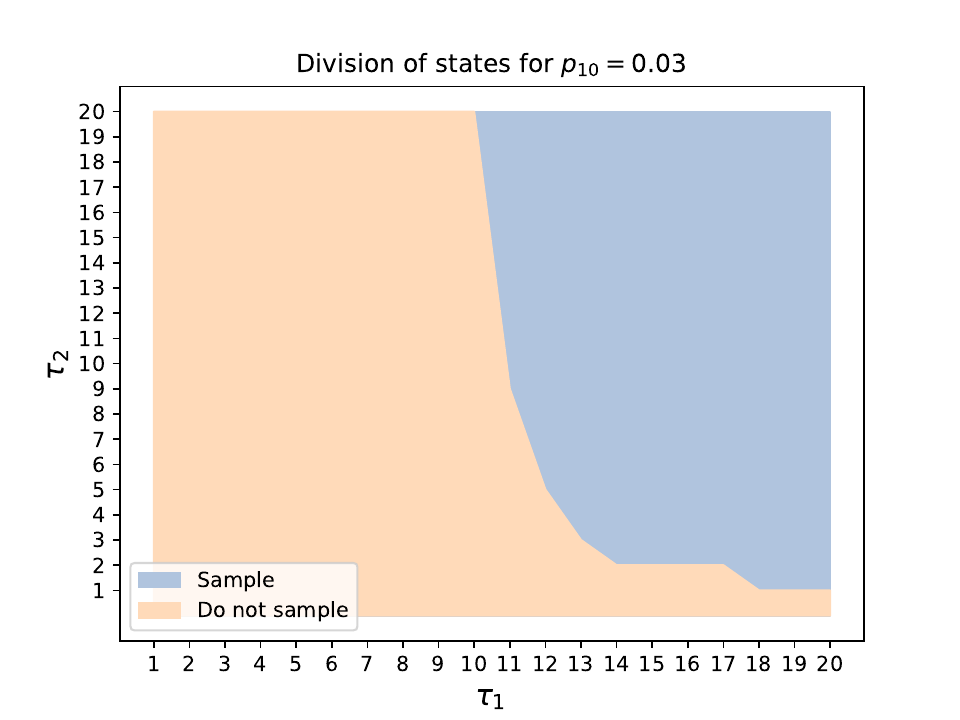}
    \caption{Division of states on the basis of optimal decision for each state for $p_{01}=0.03$, $p_{10}=0.01$, $q = 0.8$ and $\nu=0.1$.}
    \label{fig:p=0.03}
\end{figure}

    % \item {\underline{Case-3}}: $p_{01} = 0.04$, $p_{10} = 0.01$. For these values of probabilities the optimal decision is to sample for states with values of $\tau_1$ and $\tau_2$ shown in blue color in \Cref{fig:p=0.04}.

\begin{figure}
    \centering
    \includegraphics[scale=0.53]{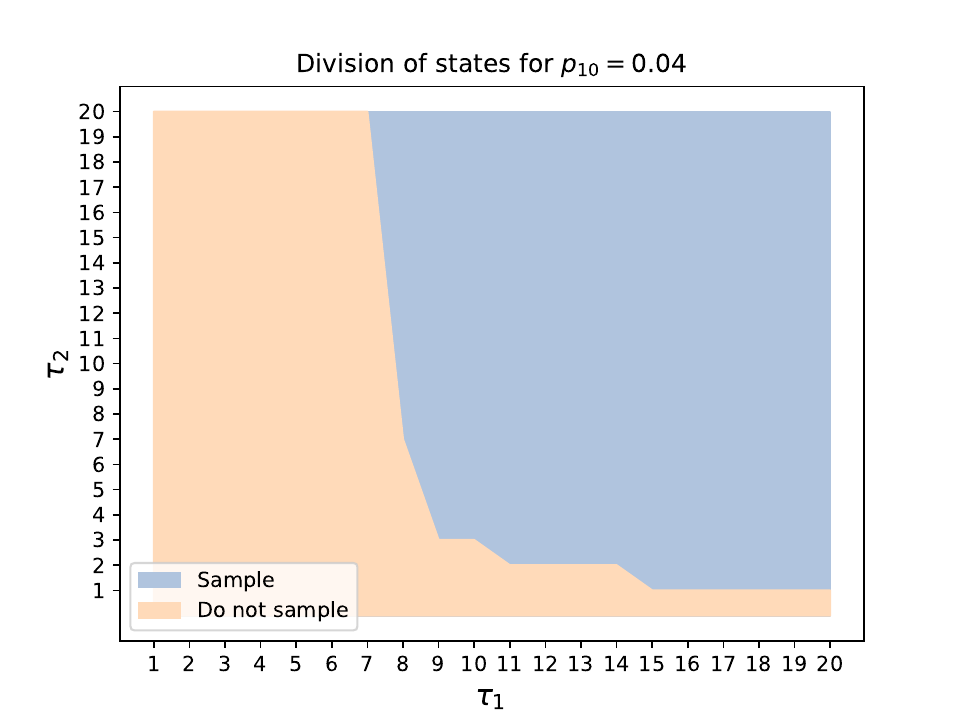}
    \caption{Division of states on the basis of optimal decision for each state for $p_{01}=0.04$, $p_{10}=0.01$, $q = 0.8$ and $\nu=0.1$.}
    \label{fig:p=0.04}
\end{figure}
% \end{itemize}

\textcolor{black}{For $p_{10} = 0.01$, we compute the optimal decision for sample or not to sample for all pairs $\tau_1$ and $\tau_2$ in the chosen range. They are then plotted with in Figures \ref{fig:p=0.02}, \ref{fig:p=0.03}, and \ref{fig:p=0.04} for $p_{01}$ values $0.02$, $0.03$, and $0.04$, respectively. The area in the blue region shows the $(\tau_1, \tau_2)$ for which the optimal decision is to sample.}
\textcolor{black}{We observe that there is a threshold structure for values of $\tau_1$ and $\tau_2$ given different values of transition probabilities, i.e., it is possible to find a $(\tau_1^*, \tau_2^*)$, such that for all $\tau_1 \geq$ $\tau_1^*$ and $\tau_2 \geq$ $\tau_2^*$, the optimal decision would be to sample. As transition probability $p_{01}$ increases, the thresholds for $\tau_1$ and $\tau_2$ reduces. This is consistent with the fact that if the latest received state is $0$, as the transition probability to state $1$ increases, a smaller sampling interval warrants freshness of detecting the state transition. This demonstrates the efficacy of AoD minimizing policy.}

The average age penalty for varying transition probability $p_{01}$ and $p_{10} = 0.01$ is shown in \Cref{fig:avg_age}. We observe that as $p_{01}$ increases, the average \metric{} decreases. This is because as $p_{01}$ increases, the steady state probability of state $1$ increases compared to state $0$, which means the DTMC moves out of state $1$ less frequently, and therefore, an optimal sampling policy achieves lower average AoD. In \Cref{fig: age_vs_samp}, we show average AoD versus the upper bound on the average sampling frequency. As expected, \metric{} decreases with the upper bound. 
%this is because as we sample more number of times the state detection is more accurate and hence a lower average \metric. \\[4pt]
\begin{figure}
    \centering
    \includegraphics[scale=0.45]{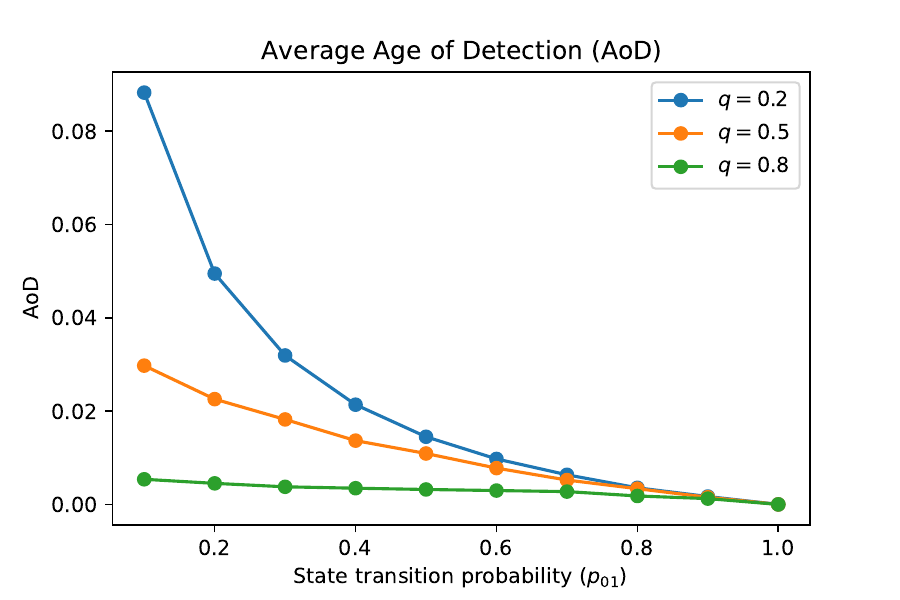}
    \caption{Variation of the average AoD for different values of state transition probability $p_{01}$ at three different values of $q$ for $p_{10}=0.01$, $\nu = 0.1$.}
    \label{fig:avg_age}
\end{figure}

% \begin{figure}
%     \centering
%     \includegraphics{samp_freq_with_p_v1.png}
%     \caption{Sampling frequency for values of $p_{01}$}
%     \label{fig:avg_samp}
% \end{figure}

 \begin{figure}
     \centering
     \includegraphics[scale=0.45]{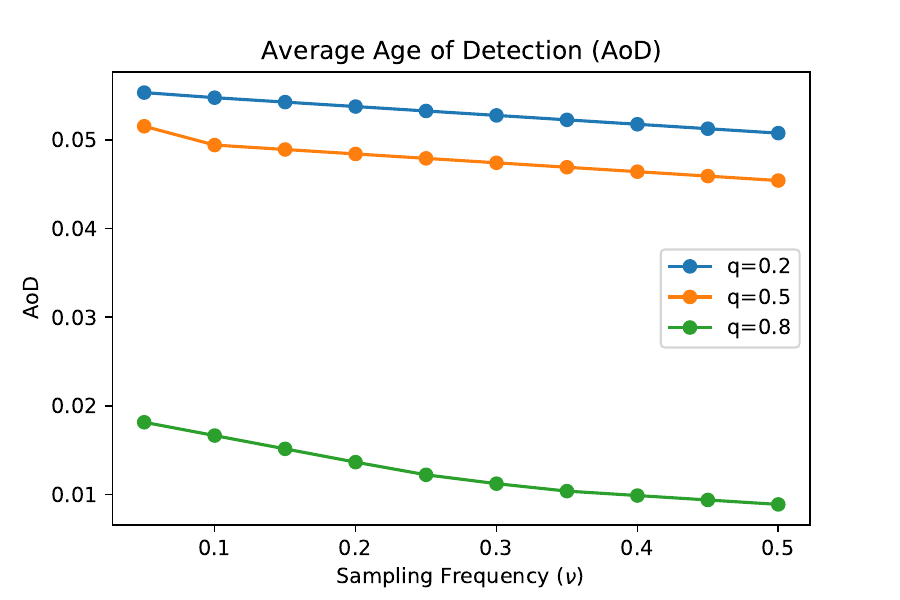}
     \caption{Variation of the average AoD for different values of Sampling frequency $(\nu)$ at three different values of $q$ for $p_{01} = 0.03$ and $p_{10} = 0.01$.}
     \label{fig: age_vs_samp}
 \end{figure}

\subsubsection{Variation with transmission success probability $q$} In \Cref{fig:avg_age_with_q}, we show the average \metric{} for varying $q$. We observe that with an increase in $q$, there is a sharp decline in the average \metric. This is because as $q$ increases, the number of times we need to transmit a packet for success decreases, resulting in more accurate state detection at the receiver, decreasing the average \metric. 

\begin{figure}
    \centering
    \includegraphics[scale=0.45]{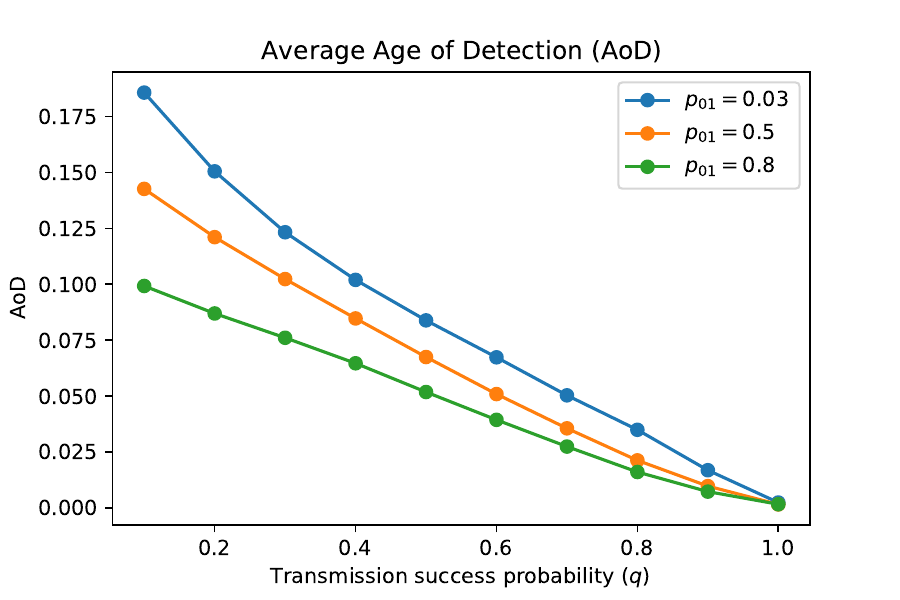}
    \caption{Variation of the average AoD for different values of transmission success probability $(q)$ at three different values of $p_{01}$ for $p_{10}=0.01$, $\nu = 0.1.$}
    \label{fig:avg_age_with_q}
\end{figure}

\subsubsection{Comparison of estimation error under \metric{} and AoI} In this subsection, we compare the average estimation error achieved under AoD and AoI minimization policies. In a given time slot, the error is equal to $1$ if the freshest state available at the receiver is not equal to the actual state of the DTMC, and is $0$, otherwise. For computing AoI minimization policy, we use the algorithm from~\cite{sun2017update}. For the parameters we have chosen, the AoI minimization policy turns out to be the zero-wait policy, i.e., sample and send immediately after every successful transmission. Hence, for this policy, the average sampling frequency trivially equals the transmission success probability $q$. For \metric, we set different sampling frequency constraints $\nu$ to calculate the estimation error. From Figure \ref{fig:AoD_vs_AoI_0.8}, we observe that for a stringent frequency constraint  $\nu = 0.2$ in the AoD minimization problem, the AoI-based estimation error is lower than that of \metric{} for all values of $p_{10}$ because the AoI sampling frequency is $0.8$. However, for $\nu = 0.6$ and $\nu = 0.8$, the \metric{} estimation error is lower than that of AoI. Thus, using \metric{} we can obtain a lower estimation error even at a lower sampling frequency $\nu =0.6$ compared to that of AoI. 
%This demonstrates the advantage of using AoD for measuring freshness for detecting DTMC state transitions.
%In \Cref{fig:AoD_vs_AoI_0.8}, we can observe that for $\nu =0.6$, we obtain a lower estimation error even compared to AoI having average sampling frequency of $0.8$. A similar trend can be seen in \Cref{fig:AoD_vs_AoI_0.7}. 
Also, from Figures \ref{fig:AoD_vs_AoI_0.8}, we observe that, as we increase $p_{10}$ the steady state probability of state $1$ decreases, which results in missing the detection of state $1$ under AoI minimizing policy and thus resulting in an increase in the estimation error with $p_{01}$. This is not the case with \metric{} minimizing policy because the cost is defined to take into account the last observed state. Further, the estimation error under the \metric{} minimizing policy decreases with an increase in the steady state probability of one state compared to the other. This shows that \metric{} is a better metric as it reduces the number of samples to achieve the same estimation error. 

Further, in Figure \ref{fig:AoD_vs_AoI_0.8_MAP}, we present the MAP estimation error, which is computed by the receiver using the MAP estimate of the current state instead of using the freshest state to compute the estimation error. Although less pronounced, we again observe that \metric{} minimizing policy archives lower MAP estimation error. These results demonstrate the advantage of AoD for measuring freshness for detecting DTMC state transitions. 

% \textcolor{red}{(Note: This explains what we are trying to do; need to add explanation for the nature of plot.)}

\begin{figure}
    \centering
    \includegraphics[scale=0.45]{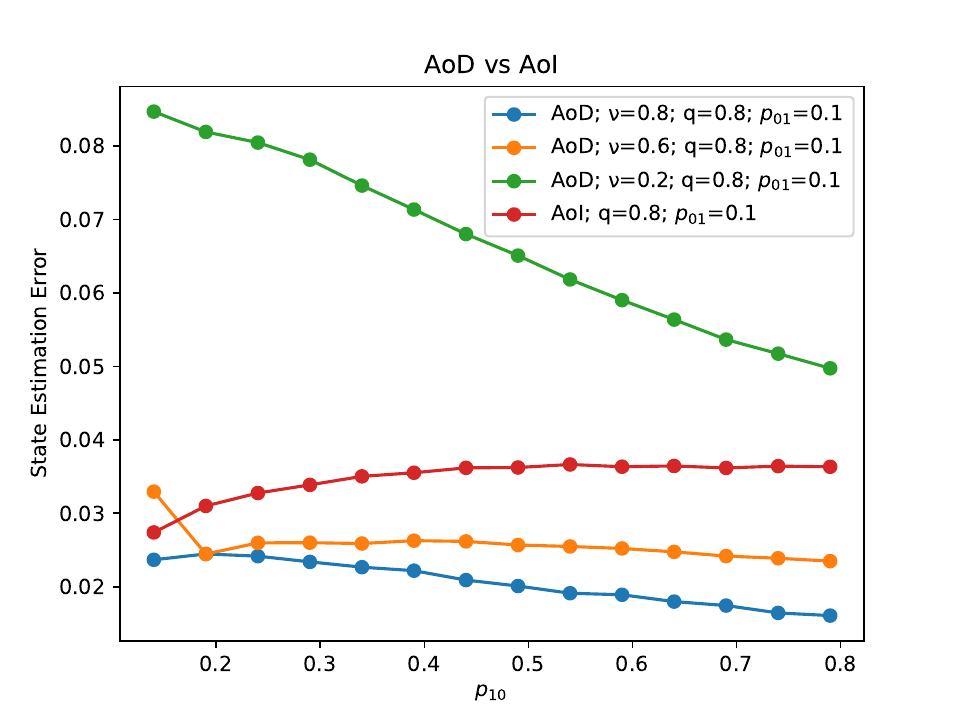}
    \caption{Comparison of state estimation error for \metric{} and AoI for different values of sampling frequency constraint $(\nu)$ for $q=0.8$ and $p_{01}=0.1$.}
    \label{fig:AoD_vs_AoI_0.8}
\end{figure}

\begin{comment}
\begin{figure}
    \centering
    \includegraphics[scale=0.45]{q=0.7_AoDvsAoI.png}
    \caption{Comparison of state estimation error for AoD and AoI for different values of sampling frequency constraint $(\nu)$ for $q=0.7$ and $p_{01}=0.1$.}
    \label{fig:AoD_vs_AoI_0.7}
\end{figure}
\end{comment}

\begin{figure}
    \centering
    \includegraphics[scale=0.45]{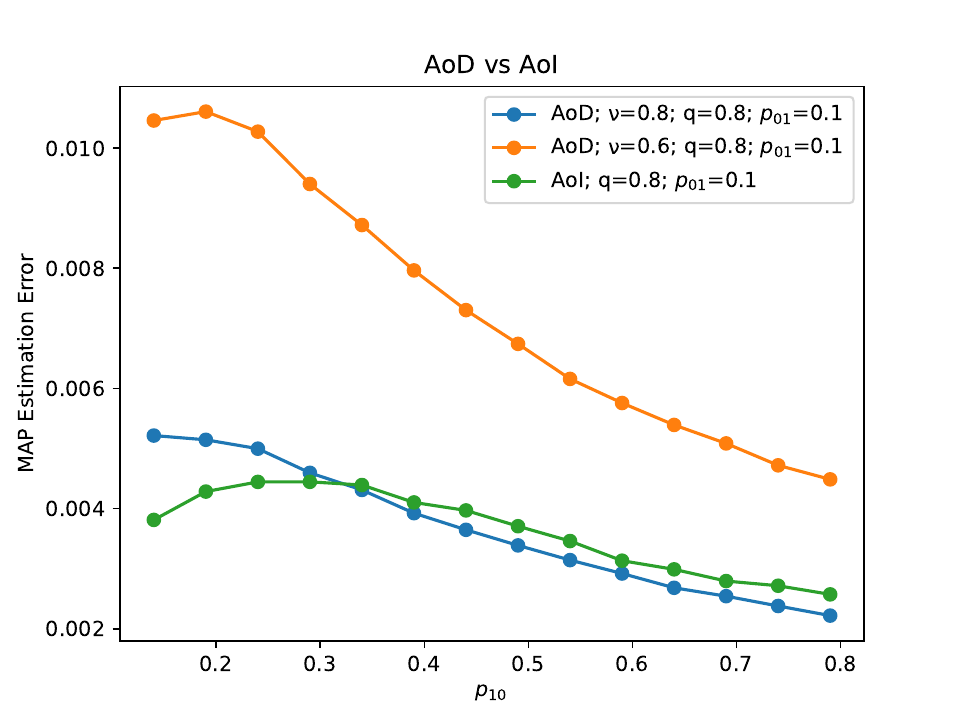}
    \caption{Comparison of MAP estimation error for \metric{} and AoI for different values of sampling frequency constraint $(\nu)$ for $q=0.8$ and $p_{01}=0.1$.}
    \label{fig:AoD_vs_AoI_0.8_MAP}
\end{figure}

\section{Conclusion}
In this work, we have proposed a novel freshness metric, Age of Detection (\metric), for detecting transitions of a DTMC source sending updates over a lossy channel. In contrast to the classical AoI metric, \metric{} takes into account the state of the DTMC, and we have shown that it extends the existing metrics age penalty and Ao$^2$I for lossy channels. Given a sampling frequency constraint, the \metric{} minimization problem is formulated as a CMDP, which is solved using the Lagrangian approach. Using simulation, we have demonstrated that the estimation error under the \metric{} minimizing policy -- sampling at a lower frequency in some settings -- is lower than that of the AoI minimizing policy. 

\bibliography{refs}
\bibliographystyle{ieeetr}
\end{document}